\documentclass{article}%
\usepackage{amssymb}
\usepackage{amsmath}
\usepackage{amsfonts}
\usepackage{graphicx}%
\providecommand{\U}[1]{\protect\rule{.1in}{.1in}}
\newtheorem{theorem}{Theorem}[section]

\newtheorem{conjecture}[theorem]{Conjecture}
\newtheorem{corollary}[theorem]{Corollary}

\newtheorem{lemma}[theorem]{Lemma}

\newtheorem{problem}[theorem]{Problem}
\newtheorem{proposition}[theorem]{Proposition}
\newtheorem{remark}[theorem]{Remark}

\newenvironment{proof}[1][Proof]{\noindent\textbf{#1.} }{\ \rule{0.5em}{0.5em}}
\begin{document}

\title{Critical and Maximum Independent Sets of a Graph}
\author{Adi Jarden\\Department of Mathematics and Computer Science\\Ariel University, Israel\\jardena@ariel.ac.il
\and Vadim E. Levit\\Department of Mathematics and Computer Science\\Ariel University, Israel\\levitv@ariel.ac.il
\and Eugen Mandrescu\\Department of Computer Science\\Holon Institute of Technology, Israel\\eugen\_m@hit.ac.il}
\date{}
\maketitle

\begin{abstract}
Let $G$ be a simple graph with vertex set $V\left(  G\right)  $. A set
$S\subseteq V\left(  G\right)  $ is \textit{independent} if no two vertices
from $S$ are adjacent. By $\mathrm{Ind}(G)$ we mean the family of all
independent sets of $G$, while $\mathrm{core}\left(  G\right)  $ and
$\mathrm{corona}\left(  G\right)  $ denote the intersection and the union of
all maximum independent sets, respectively.

The number $d\left(  X\right)  =$ $\left\vert X\right\vert -\left\vert
N(X)\right\vert $ is the \textit{difference} of $X\subseteq V\left(  G\right)
$, and a set $A\in\mathrm{Ind}(G)$ is \textit{critical} if $d(A)=\max
\{d\left(  I\right)  :I\in\mathrm{Ind}(G)\}$ \cite{Zhang1990}.

Let $\mathrm{\ker}(G)$ and $\mathrm{diadem}(G)$ be the intersection and union,
respectively, of all critical independent sets of $G$ \cite{LevMan2012a}.

In this paper, we present various connections between critical unions and
intersections of maximum independent sets of a graph. These relations give
birth to new characterizations of K\"{o}nig-Egerv\'{a}ry graphs, some of them
involving $\mathrm{\ker}(G)$, $\mathrm{core}\left(  G\right)  $,
$\mathrm{corona}\left(  G\right)  $, and $\mathrm{diadem}(G)$.


\end{abstract}

\section{Introduction}

Throughout this paper $G$ is a finite simple graph with vertex set $V(G)$ and
edge set $E(G)$. If $X\subseteq V\left(  G\right)  $, then $G[X]$ is the
subgraph of $G$ induced by $X$. By $G-W$ we mean either the subgraph
$G[V\left(  G\right)  -W]$, if $W\subseteq V(G)$, or the subgraph obtained by
deleting the edge set $W$, for $W\subseteq E(G)$. In either case, we use
$G-w$, whenever $W$ $=\{w\}$. If $A,B$ $\subseteq V\left(  G\right)  $, then
$(A,B)$ stands for the set $\{ab:a\in A,b\in B,ab\in E\left(  G\right)  \}$.

The \textit{neighborhood} $N(v)$ of $v\in V\left(  G\right)  $ is the set
$\{w:w\in V\left(  G\right)  $ \textit{and} $vw\in E\left(  G\right)  \}$. In
order to avoid ambiguity, we use also $N_{G}(v)$ instead of $N(v)$. The
\textit{neighborhood} $N(A)$ of $A\subseteq V\left(  G\right)  $ is $\{v\in
V\left(  G\right)  :N(v)\cap A\neq\emptyset\}$, and $N[A]=N(A)\cup A$.

A set $S\subseteq V(G)$ is \textit{independent} if no two vertices from $S$
are adjacent, and by $\mathrm{Ind}(G)$ we mean the family of all the
independent sets of $G$. An independent set of maximum size is a
\textit{maximum independent set} of $G$, and $\alpha(G)=\max\{\left\vert
S\right\vert :S\in\mathrm{Ind}(G)\}$.

\begin{theorem}
\label{Berge}\cite{Berge1981}, \cite{Berge1985} An independent set $X$ is
maximum if and only if every independent set $S$ disjoint from $X$ can be
matched into $X$.
\end{theorem}

For a graph $G$, let $\Omega(G)$ denote the family of all its maximum
independent sets, $\mathrm{core}(G)=%
%TCIMACRO{\dbigcap }%
%BeginExpansion
{\displaystyle\bigcap}
%EndExpansion
\{S:S\in\Omega(G)\}$ \cite{LevMan2002a}, and $\mathrm{corona}(G)=%
%TCIMACRO{\dbigcup }%
%BeginExpansion
{\displaystyle\bigcup}
%EndExpansion
\{S:S\in\Omega(G)\}$ \cite{BorosGolLev}.

It is clear that $N\left(  \mathrm{core}(G)\right)  \subseteq$ $V\left(
G\right)  -\mathrm{corona}(G)$, and there exist graphs satisfying $N\left(
\mathrm{core}(G)\right)  \neq$ $V\left(  G\right)  -\mathrm{corona}(G)$ (for
some examples, see the graphs from Figure \ref{fig101101}, where
\textrm{core}$(G_{1})=\{a,b\}$ and \textrm{core}$(G_{2})=\{x,y,z\}$).

The problem of whether $\mathrm{core}(G)\neq\emptyset$ is \textbf{NP}-hard
\cite{BorosGolLev}.\begin{figure}[h]
\setlength{\unitlength}{1.0cm} \begin{picture}(5,1.9)\thicklines
\multiput(2,0.5)(1,0){5}{\circle*{0.29}}
\multiput(2,1.5)(1,0){5}{\circle*{0.29}}
\put(2,0.5){\line(1,0){4}}
\put(2,1.5){\line(1,-1){1}}
\put(3,1.5){\line(1,-1){1}}
\put(3,1.5){\line(1,0){1}}
\put(4,0.5){\line(0,1){1}}
\put(4,0.5){\line(1,1){1}}
\put(5,1.5){\line(1,0){1}}
\put(6,0.5){\line(0,1){1}}
\put(2,0.1){\makebox(0,0){$a$}}
\put(2,1.8){\makebox(0,0){$b$}}
\put(3,0.1){\makebox(0,0){$c$}}
\put(4,0.1){\makebox(0,0){$d$}}
\put(5,0.1){\makebox(0,0){$e$}}
\put(6,0.1){\makebox(0,0){$f$}}
\put(3,1.8){\makebox(0,0){$x$}}
\put(4,1.8){\makebox(0,0){$y$}}
\put(5,1.8){\makebox(0,0){$u$}}
\put(6,1.8){\makebox(0,0){$v$}}
\put(1.3,1){\makebox(0,0){$G_{1}$}}
\multiput(8,0.5)(1,0){6}{\circle*{0.29}}
\multiput(8,1.5)(1,0){6}{\circle*{0.29}}
\put(8,0.5){\line(1,1){1}}
\put(8,1.5){\line(1,0){1}}
\put(9,0.5){\line(1,0){4}}
\put(9,0.5){\line(0,1){1}}
\put(10,0.5){\line(0,1){1}}
\put(10,1.5){\line(1,0){1}}
\put(11,1.5){\line(1,-1){1}}
\put(12,1.5){\line(1,0){1}}
\put(12,1.5){\line(1,-1){1}}
\put(13,0.5){\line(0,1){1}}
\put(8,1.8){\makebox(0,0){$x$}}
\put(8,0.1){\makebox(0,0){$y$}}
\put(9,0.1){\makebox(0,0){$z$}}
\put(9,1.8){\makebox(0,0){$v_{1}$}}
\put(10,1.8){\makebox(0,0){$v_{2}$}}
\put(11,1.8){\makebox(0,0){$v_{3}$}}
\put(12,1.8){\makebox(0,0){$v_{7}$}}
\put(13,1.8){\makebox(0,0){$v_{8}$}}
\put(10,0.11){\makebox(0,0){$v_{6}$}}
\put(11,0.11){\makebox(0,0){$v_{5}$}}
\put(12,0.11){\makebox(0,0){$v_{4}$}}
\put(13,0.11){\makebox(0,0){$v_{9}$}}
\put(7.3,1){\makebox(0,0){$G_{2}$}}
\end{picture}\caption{$V(G_{1})=$ \textrm{corona}$(G_{1})\cup N\left(
\mathrm{core}(G_{1})\right)  \cup\{d\}$, $V(G_{2})=$ \textrm{corona}%
$(G_{2})\cup N\left(  \mathrm{core}(G_{2})\right)  $.}%
\label{fig101101}%
\end{figure}
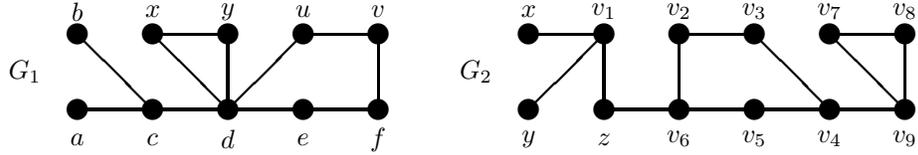

A \textit{matching} is a set $M$ of pairwise non-incident edges of $G$. If
$A\subseteq V(G)$, then $M\left(  A\right)  $ is the set of all the vertices
matched by $M$ with vertices belonging to $A$. A matching of maximum
cardinality, denoted $\mu(G)$, is a \textit{maximum matching}.

\begin{lemma}
[Matching Lemma]\cite{LevManLemma2011} \label{MatchLem}If $A\in\mathrm{Ind}%
(G),\Lambda\subseteq\Omega(G)$, and $\left\vert \Lambda\right\vert \geq1$,
then there exists a matching from $A-%
%TCIMACRO{\dbigcap }%
%BeginExpansion
{\displaystyle\bigcap}
%EndExpansion
\Lambda$ into $%
%TCIMACRO{\dbigcup }%
%BeginExpansion
{\displaystyle\bigcup}
%EndExpansion
\Lambda-A$.
\end{lemma}

For $X\subseteq V(G)$, the number $\left\vert X\right\vert -\left\vert
N(X)\right\vert $ is the \textit{difference} of $X$, denoted $d(X)$. The
\textit{critical difference} $d(G)$ is $\max\{d(X):X\subseteq V(G)\}$. The
number $\max\{d(I):I\in\mathrm{Ind}(G)\}$ is the \textit{critical independence
difference} of $G$, denoted $id(G)$. Clearly, $d(G)\geq id(G)$. It was shown
in \cite{Zhang1990} that $d(G)$ $=id(G)$ holds for every graph $G$. If $A$ is
an independent set in $G$ with $d\left(  X\right)  =id(G)$, then $A$ is a
\textit{critical independent set} \cite{Zhang1990}.

For example, consider the graph $G$ of Figure \ref{fig511}, where
$X=\{v_{1},v_{2},v_{3},v_{4}\}$ is a critical set, while $I=\{v_{1}%
,v_{2},v_{3},v_{6},v_{7}\}$ is a critical independent set. Other critical sets
are $\{v_{1},v_{2}\}$, $\{v_{1},v_{2},v_{3}\}$, $\{v_{1},v_{2},v_{3}%
,v_{4},v_{6},v_{7}\}$. 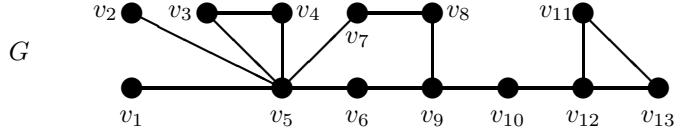
\begin{figure}[h]
\setlength{\unitlength}{1cm}\begin{picture}(5,1.9)\thicklines
\multiput(6,0.5)(1,0){6}{\circle*{0.29}}
\multiput(5,1.5)(1,0){4}{\circle*{0.29}}
\multiput(4,0.5)(0,1){2}{\circle*{0.29}}
\put(10,1.5){\circle*{0.29}}
\put(4,0.5){\line(1,0){7}}
\put(4,1.5){\line(2,-1){2}}
\put(5,1.5){\line(1,-1){1}}
\put(5,1.5){\line(1,0){1}}
\put(6,0.5){\line(0,1){1}}
\put(6,0.5){\line(1,1){1}}
\put(7,1.5){\line(1,0){1}}
\put(8,0.5){\line(0,1){1}}
\put(10,0.5){\line(0,1){1}}
\put(10,1.5){\line(1,-1){1}}
\put(4,0.1){\makebox(0,0){$v_{1}$}}
\put(3.65,1.5){\makebox(0,0){$v_{2}$}}
\put(4.65,1.5){\makebox(0,0){$v_{3}$}}
\put(6.35,1.5){\makebox(0,0){$v_{4}$}}
\put(6,0.1){\makebox(0,0){$v_{5}$}}
\put(7,0.1){\makebox(0,0){$v_{6}$}}
\put(7,1.15){\makebox(0,0){$v_{7}$}}
\put(8,0.1){\makebox(0,0){$v_{9}$}}
\put(8.35,1.5){\makebox(0,0){$v_{8}$}}
\put(9.65,1.5){\makebox(0,0){$v_{11}$}}
\put(9,0.1){\makebox(0,0){$v_{10}$}}
\put(10,0.1){\makebox(0,0){$v_{12}$}}
\put(11,0.1){\makebox(0,0){$v_{13}$}}
\put(2.5,1){\makebox(0,0){$G$}}
\end{picture}\caption{\textrm{core}$(G)=\{v_{1},v_{2},v_{6},v_{10}\}$ is a
critical set, since $d\left(  \mathrm{core}(G)\right)  =1=d\left(  G\right)
$.}%
\label{fig511}%
\end{figure}

It is known that finding a maximum independent set is an \textbf{NP}-hard
problem \cite{GaryJohnson79}. Zhang proved that a critical independent set can
be found in polynomial time \cite{Zhang1990}.

\begin{theorem}
\label{th3}\cite{ButTruk2007} Each critical independent set is included in a
maximum independent set.
\end{theorem}

Theorem \ref{th3} leads to an efficient way of approximating $\alpha(G)$
\cite{Truchanov2008}. Moreover, every critical independent set is contained in
a maximum critical independent set, and such a maximum critical independent
set can be found in polynomial time \cite{Larson2007}.

\begin{theorem}
\label{th2}\cite{Larson2007} There is a matching from $N(S)$ into $S$ for
every critical independent set $S$.
\end{theorem}

It is well-known that $\alpha(G)+\mu(G)\leq\left\vert V(G)\right\vert $ holds
for every graph $G$. Recall that if $\alpha(G)+\mu(G)=\left\vert
V(G)\right\vert $, then $G$ is a \textit{K\"{o}nig-Egerv\'{a}ry graph}
\cite{Deming1979,Sterboul1979}. For example, each bipartite graph is a
K\"{o}nig-Egerv\'{a}ry graph as well. Various properties of
K\"{o}nig-Egerv\'{a}ry graphs can be found in
\cite{Korach2006,levm4,LevMan2013b}. It turns out that K\"{o}nig-Egerv\'{a}ry
graphs are exactly the graphs having a critical maximum independent set
\cite{Larson2011}. In \cite{LevMan2012b} it was shown the following.

\begin{lemma}
\cite{LevMan2012b} \label{lem2} $d(G)=\alpha(G)-\mu(G)$ holds for each
K\"{o}nig-Egerv\'{a}ry graph $G$.
\end{lemma}

Using this finding, we have strengthened the characterization from
\cite{Larson2011}.

\begin{theorem}
\label{th5}\cite{LevMan2012b} For a graph $G$, the following assertions are equivalent:

\emph{(i)} $G$ is a K\"{o}nig-Egerv\'{a}ry graph;

\emph{(ii)} there exists some maximum independent set which is critical;

\emph{(iii)} each of its maximum independent sets is critical.
\end{theorem}

For a graph $G$, let $\mathrm{\ker}(G)$ be the intersection of all its
critical independent sets \cite{LevMan2012a}, and $\mathrm{diadem}(G)=%
%TCIMACRO{\dbigcup }%
%BeginExpansion
{\displaystyle\bigcup}
%EndExpansion
\left\{  S:S\text{ \textit{is a critical independent set}}\right\}  $.

In this paper we present several properties of critical unions and
intersections of maximum independent sets leading to new characterizations of
K\"{o}nig-Egerv\'{a}ry graphs, in terms of $\mathrm{core}(G)$,
$\mathrm{corona}(G)$, and $\mathrm{diadem}(G)$.

\section{Preliminaries}

Let $G$ be the graph from Figure \ref{fig511}; the sets $X=\left\{
v_{1},v_{2},v_{3}\right\}  $, $Y=\left\{  v_{1},v_{2},v_{4}\right\}  $\ are
critical independent, and the sets $X\cap Y$, $X\cup Y$ are also critical, but
only $X\cap Y$ is also independent.\ In addition, one can easily see that
$\mathrm{\ker}(G)=\left\{  v_{1},v_{2}\right\}  \subseteq\mathrm{core}(G)$,
and $\mathrm{\ker}(G)$ is a minimal critical independent set of $G$.

\begin{theorem}
\label{th4}\cite{LevMan2012a} For a graph $G$, the following assertions are true:

\emph{(i)} $\mathrm{\ker}(G)\subseteq\mathrm{core}(G)$;

\emph{(ii)} if $A$ and $B$ are critical in $G$, then $A\cup B$ and $A\cap B$
are critical as well;

\emph{(iii)} $G$ has a unique minimal independent critical set, namely,
$\mathrm{\ker}(G)$.
\end{theorem}

Various properties of $\mathrm{\ker}(G)$ and $\mathrm{core}(G)$ can be found
in \cite{LevMan2012c,LevMan2013a,LevMan2014}.

As an immediate consequence of Theorem \ref{th4}, we have the following.

\begin{corollary}
\label{cor3}For every graph $G$, $\mathrm{diadem}(G)$ is a critical set.
\end{corollary}

For instance, the graph $G$ from Figure \ref{fig511} has $\mathrm{diadem}%
(G)=\left\{  v_{1},v_{2},v_{3},v_{4},v_{6},v_{7},v_{10}\right\}  $, which is
critical, but not independent.

The graph $G_{1}$ from Figure \ref{fig101101} has $d\left(  G_{1}\right)  =1$
and $d\left(  \mathrm{corona}(G_{1})\right)  =0$, which means that
$\mathrm{corona}(G_{1})$ is not a critical set. Notice that $G_{1}$ is not a
K\"{o}nig-Egerv\'{a}ry graph. Combining Theorems \ref{th5} and \ref{th4}%
\emph{(ii)}, we deduce the following.

\begin{corollary}
\label{cor2}If $G$ is a K\"{o}nig-Egerv\'{a}ry graph, then both $\mathrm{core}%
(G)$ and $\mathrm{corona}(G)$ are critical sets. Moreover, $\mathrm{corona}%
(G)=%
%TCIMACRO{\dbigcup }%
%BeginExpansion
{\displaystyle\bigcup}
%EndExpansion
\left\{  A:A\text{ \textit{is a maximum critical independent set}}\right\}  $.
\end{corollary}

The converse of Corollary \ref{cor2} is not necessarily true; e.g., the graph
$G_{2}$ in Figure \ref{fig101101} is not a K\"{o}nig-Egerv\'{a}ry graph, while
$\mathrm{core}(G_{2})$ and $\mathrm{corona}(G_{2})$\ are critical.

\section{Unions and intersections of maximum independent sets}

\begin{theorem}
\label{Th1}Let $\Lambda\subseteq\Omega(G)$, and $\left\vert \Lambda\right\vert
\geq1$. Then%
\[
d\left(
%TCIMACRO{\dbigcup }%
%BeginExpansion
{\displaystyle\bigcup}
%EndExpansion
\Lambda\right)  =\left\vert
%TCIMACRO{\dbigcap }%
%BeginExpansion
{\displaystyle\bigcap}
%EndExpansion
\Lambda\right\vert +\left\vert
%TCIMACRO{\dbigcup }%
%BeginExpansion
{\displaystyle\bigcup}
%EndExpansion
\Lambda\right\vert -\left\vert V\left(  G\right)  \right\vert \geq\max
_{S\in\Lambda}d\left(  S\right)  \text{.}%
\]
In particular,
\[
d\left(  \mathrm{corona}(G)\right)  =\left\vert \mathrm{corona}(G)\right\vert
+\left\vert \mathrm{core}(G)\right\vert -\left\vert V\left(  G\right)
\right\vert \geq2\alpha\left(  G\right)  -\left\vert V\left(  G\right)
\right\vert =\max_{S\in\Omega\left(  G\right)  }d\left(  S\right)  \text{.}%
\]

\end{theorem}

\begin{proof}
Every vertex in $%
%TCIMACRO{\dbigcup }%
%BeginExpansion
{\displaystyle\bigcup}
%EndExpansion
\Lambda-%
%TCIMACRO{\dbigcap }%
%BeginExpansion
{\displaystyle\bigcap}
%EndExpansion
\Lambda$ has a neighbor in $%
%TCIMACRO{\dbigcup }%
%BeginExpansion
{\displaystyle\bigcup}
%EndExpansion
\Lambda-%
%TCIMACRO{\dbigcap }%
%BeginExpansion
{\displaystyle\bigcap}
%EndExpansion
\Lambda$, since $\Lambda\subseteq\Omega(G)$. Therefore, $N\left(
%TCIMACRO{\dbigcup }%
%BeginExpansion
{\displaystyle\bigcup}
%EndExpansion
\Lambda\right)  =\left(
%TCIMACRO{\dbigcup }%
%BeginExpansion
{\displaystyle\bigcup}
%EndExpansion
\Lambda-%
%TCIMACRO{\dbigcap }%
%BeginExpansion
{\displaystyle\bigcap}
%EndExpansion
\Lambda\right)  \cup\left(  V\left(  G\right)  -%
%TCIMACRO{\dbigcup }%
%BeginExpansion
{\displaystyle\bigcup}
%EndExpansion
\Lambda\right)  $, which implies%
\begin{gather*}
d\left(
%TCIMACRO{\dbigcup }%
%BeginExpansion
{\displaystyle\bigcup}
%EndExpansion
\Lambda\right)  =\left\vert
%TCIMACRO{\dbigcup }%
%BeginExpansion
{\displaystyle\bigcup}
%EndExpansion
\Lambda\right\vert -\left\vert N\left(
%TCIMACRO{\dbigcup }%
%BeginExpansion
{\displaystyle\bigcup}
%EndExpansion
\Lambda\right)  \right\vert =\left\vert
%TCIMACRO{\dbigcup }%
%BeginExpansion
{\displaystyle\bigcup}
%EndExpansion
\Lambda\right\vert -\left\vert \left(
%TCIMACRO{\dbigcup }%
%BeginExpansion
{\displaystyle\bigcup}
%EndExpansion
\Lambda-%
%TCIMACRO{\dbigcap }%
%BeginExpansion
{\displaystyle\bigcap}
%EndExpansion
\Lambda\right)  \cup\left(  V\left(  G\right)  -%
%TCIMACRO{\dbigcup }%
%BeginExpansion
{\displaystyle\bigcup}
%EndExpansion
\Lambda\right)  \right\vert =\\
=\left\vert
%TCIMACRO{\dbigcap }%
%BeginExpansion
{\displaystyle\bigcap}
%EndExpansion
\Lambda\right\vert -\left(  \left\vert V\left(  G\right)  \right\vert
-\left\vert
%TCIMACRO{\dbigcup }%
%BeginExpansion
{\displaystyle\bigcup}
%EndExpansion
\Lambda\right\vert \right)  =\left\vert
%TCIMACRO{\dbigcap }%
%BeginExpansion
{\displaystyle\bigcap}
%EndExpansion
\Lambda\right\vert +\left\vert
%TCIMACRO{\dbigcup }%
%BeginExpansion
{\displaystyle\bigcup}
%EndExpansion
\Lambda\right\vert -\left\vert V\left(  G\right)  \right\vert .
\end{gather*}
On the other hand, for every $S\in\Omega\left(  G\right)  $ we have
\[
d\left(  S\right)  =\alpha\left(  G\right)  -\left(  \left\vert V\left(
G\right)  \right\vert -\alpha\left(  G\right)  \right)  =2\alpha\left(
G\right)  -\left\vert V\left(  G\right)  \right\vert .
\]

Since $\left\vert
%TCIMACRO{\dbigcap }%
%BeginExpansion
{\displaystyle\bigcap}
%EndExpansion
\Lambda\right\vert +\left\vert
%TCIMACRO{\dbigcup }%
%BeginExpansion
{\displaystyle\bigcup}
%EndExpansion
\Lambda\right\vert \geq2\alpha\left(  G\right)  $, we obtain%
\[
d\left(
%TCIMACRO{\dbigcup }%
%BeginExpansion
{\displaystyle\bigcup}
%EndExpansion
\Lambda\right)  =\left\vert
%TCIMACRO{\dbigcap }%
%BeginExpansion
{\displaystyle\bigcap}
%EndExpansion
\Lambda\right\vert +\left\vert
%TCIMACRO{\dbigcup }%
%BeginExpansion
{\displaystyle\bigcup}
%EndExpansion
\Lambda\right\vert -\left\vert V\left(  G\right)  \right\vert \geq
2\alpha\left(  G\right)  -\left\vert V\left(  G\right)  \right\vert =d\left(
S\right)  \text{,}%
\]
as required.

In particular, if $\Lambda=\Omega(G)$, then $%
%TCIMACRO{\dbigcup }%
%BeginExpansion
{\displaystyle\bigcup}
%EndExpansion
\Lambda=\mathrm{corona}(G)$, $%
%TCIMACRO{\dbigcap }%
%BeginExpansion
{\displaystyle\bigcap}
%EndExpansion
\Lambda=\mathrm{core}(G)$, and the conclusion follows.
\end{proof}

Notice that if $A$ is a critical independent set in a graph $G$ having
$d(G)>0$, then $A\cap S\neq\emptyset$ holds for every $S\in\Omega(G)$, because
$\emptyset\neq\ker\left(  G\right)  \subseteq A\cap\mathrm{core}(G)\subseteq
A\cap S$, according to Theorem \ref{th4}\emph{(i)}.

\begin{proposition}
Let $A$ be a critical independent set of a graph $G$ with $\ker\left(
G\right)  =\emptyset$, and $\Lambda=\{S\in\Omega(G):A\cap S=\emptyset\}$. Then
$\left\vert
%TCIMACRO{\dbigcap }%
%BeginExpansion
{\displaystyle\bigcap}
%EndExpansion
\Lambda\right\vert \geq\left\vert A\right\vert $.
\end{proposition}

\begin{proof}
Let $S\in\Lambda$. Since $A$ is critical and $d(G)=0$, it follows that
$\left\vert A\right\vert =\left\vert N\left(  A\right)  \right\vert $. By
Theorem \ref{Berge}, there is a matching from $A$ into $S$, because $A$ is
independent and disjoint from $S$. Consequently, we infer that $N\left(
A\right)  \subset S$. Hence, we obtain $\left\vert
%TCIMACRO{\dbigcap }%
%BeginExpansion
{\displaystyle\bigcap}
%EndExpansion
\Lambda\right\vert \geq\left\vert N(A)\right\vert =\left\vert A\right\vert $,
as required.
\end{proof}

\begin{theorem}
\label{Th2}Let $\Lambda\subseteq\Omega(G)$, and $\left\vert \Lambda\right\vert
\geq1$.

\emph{(i)} If $%
%TCIMACRO{\dbigcup }%
%BeginExpansion
{\displaystyle\bigcup}
%EndExpansion
\Lambda$ is critical, then $\left\vert N\left(
%TCIMACRO{\dbigcap }%
%BeginExpansion
{\displaystyle\bigcap}
%EndExpansion
\Lambda\right)  \right\vert +\left\vert
%TCIMACRO{\dbigcup }%
%BeginExpansion
{\displaystyle\bigcup}
%EndExpansion
\Lambda\right\vert =\left\vert V\left(  G\right)  \right\vert $, and $%
%TCIMACRO{\dbigcap }%
%BeginExpansion
{\displaystyle\bigcap}
%EndExpansion
\Lambda$ is critical.

\emph{(ii)} If $%
%TCIMACRO{\dbigcap }%
%BeginExpansion
{\displaystyle\bigcap}
%EndExpansion
\Lambda$ is critical, then
\[
\left\vert N\left(
%TCIMACRO{\dbigcap }%
%BeginExpansion
{\displaystyle\bigcap}
%EndExpansion
\Lambda\right)  \right\vert +\left\vert
%TCIMACRO{\dbigcup }%
%BeginExpansion
{\displaystyle\bigcup}
%EndExpansion
\Lambda\right\vert \leq\left\vert V\left(  G\right)  \right\vert \text{, and
}d\left(
%TCIMACRO{\dbigcap }%
%BeginExpansion
{\displaystyle\bigcap}
%EndExpansion
\Lambda\right)  \geq2\alpha\left(  G\right)  -\left\vert V\left(  G\right)
\right\vert .
\]

\end{theorem}

\begin{proof}
\emph{(i)} By definition of $d(G)$ and Theorem \ref{Th1}, we get
\[
d\left(  G\right)  =d\left(
%TCIMACRO{\dbigcup }%
%BeginExpansion
{\displaystyle\bigcup}
%EndExpansion
\Lambda\right)  =\left\vert
%TCIMACRO{\dbigcap }%
%BeginExpansion
{\displaystyle\bigcap}
%EndExpansion
\Lambda\right\vert +\left\vert
%TCIMACRO{\dbigcup }%
%BeginExpansion
{\displaystyle\bigcup}
%EndExpansion
\Lambda\right\vert -\left\vert V\left(  G\right)  \right\vert \geq d\left(
%TCIMACRO{\dbigcap }%
%BeginExpansion
{\displaystyle\bigcap}
%EndExpansion
\Lambda\right)  =\left\vert
%TCIMACRO{\dbigcap }%
%BeginExpansion
{\displaystyle\bigcap}
%EndExpansion
\Lambda\right\vert -\left\vert N\left(
%TCIMACRO{\dbigcap }%
%BeginExpansion
{\displaystyle\bigcap}
%EndExpansion
\Lambda\right)  \right\vert .
\]
Hence we infer that $\left\vert N\left(
%TCIMACRO{\dbigcap }%
%BeginExpansion
{\displaystyle\bigcap}
%EndExpansion
\Lambda\right)  \right\vert +\left\vert
%TCIMACRO{\dbigcup }%
%BeginExpansion
{\displaystyle\bigcup}
%EndExpansion
\Lambda\right\vert \geq\left\vert V\left(  G\right)  \right\vert $. Thus,
$\left\vert N\left(
%TCIMACRO{\dbigcap }%
%BeginExpansion
{\displaystyle\bigcap}
%EndExpansion
\Lambda\right)  \right\vert +\left\vert
%TCIMACRO{\dbigcup }%
%BeginExpansion
{\displaystyle\bigcup}
%EndExpansion
\Lambda\right\vert =\left\vert V\left(  G\right)  \right\vert $, because
$\left(  N\left(
%TCIMACRO{\dbigcap }%
%BeginExpansion
{\displaystyle\bigcap}
%EndExpansion
\Lambda\right)  \right)  \cap\left(
%TCIMACRO{\dbigcup }%
%BeginExpansion
{\displaystyle\bigcup}
%EndExpansion
\Lambda\right)  =\emptyset$.

Moreover, we deduce that%
\[
d\left(
%TCIMACRO{\dbigcup }%
%BeginExpansion
{\displaystyle\bigcup}
%EndExpansion
\Lambda\right)  =\left\vert
%TCIMACRO{\dbigcap }%
%BeginExpansion
{\displaystyle\bigcap}
%EndExpansion
\Lambda\right\vert +\left\vert
%TCIMACRO{\dbigcup }%
%BeginExpansion
{\displaystyle\bigcup}
%EndExpansion
\Lambda\right\vert -\left\vert V\left(  G\right)  \right\vert =\left\vert
%TCIMACRO{\dbigcap }%
%BeginExpansion
{\displaystyle\bigcap}
%EndExpansion
\Lambda\right\vert -\left\vert N\left(
%TCIMACRO{\dbigcap }%
%BeginExpansion
{\displaystyle\bigcap}
%EndExpansion
\Lambda\right)  \right\vert =d\left(
%TCIMACRO{\dbigcap }%
%BeginExpansion
{\displaystyle\bigcap}
%EndExpansion
\Lambda\right)  \text{,}%
\]
i.e., $%
%TCIMACRO{\dbigcap }%
%BeginExpansion
{\displaystyle\bigcap}
%EndExpansion
\Lambda$ is a critical set.

\emph{(ii)} By definition of $d(G)$ and Theorem \ref{Th1}, we have%
\begin{gather*}
d\left(  G\right)  =d\left(
%TCIMACRO{\dbigcap }%
%BeginExpansion
{\displaystyle\bigcap}
%EndExpansion
\Lambda\right)  =\left\vert
%TCIMACRO{\dbigcap }%
%BeginExpansion
{\displaystyle\bigcap}
%EndExpansion
\Lambda\right\vert -\left\vert N\left(
%TCIMACRO{\dbigcap }%
%BeginExpansion
{\displaystyle\bigcap}
%EndExpansion
\Lambda\right)  \right\vert \geq\\
\geq d\left(
%TCIMACRO{\dbigcup }%
%BeginExpansion
{\displaystyle\bigcup}
%EndExpansion
\Lambda\right)  =\left\vert
%TCIMACRO{\dbigcap }%
%BeginExpansion
{\displaystyle\bigcap}
%EndExpansion
\Lambda\right\vert +\left\vert
%TCIMACRO{\dbigcup }%
%BeginExpansion
{\displaystyle\bigcup}
%EndExpansion
\Lambda\right\vert -\left\vert V\left(  G\right)  \right\vert \geq
2\alpha\left(  G\right)  -\left\vert V\left(  G\right)  \right\vert ,
\end{gather*}
which completes the proof.
\end{proof}

In particular, taking $\Lambda=\Omega(G)$ in Theorem \ref{Th2}, we obtain the following.

\begin{corollary}
\label{cor1}If $\mathrm{corona}(G)$ is a critical set, then $\left\vert
\mathrm{corona}(G)\right\vert +\left\vert N\left(  \mathrm{core}(G)\right)
\right\vert =\left\vert V\left(  G\right)  \right\vert $ and $\mathrm{core}%
(G)$ is critical.
\end{corollary}

Notice that if $\mathrm{core}(G)$ is critical, then $\mathrm{corona}(G)$ is
not necessarily critical. For example, the graph $G_{1}$ from Figure
\ref{fig101101} has $d\left(  G_{1}\right)  =d(\mathrm{core}(G_{1}))=1$, while
$\mathrm{corona}(G_{1})$ is not a critical set.

\begin{theorem}
\label{th9}If $G$ is a K\"{o}nig-Egerv\'{a}ry graph, then

\emph{(i)} \cite{JLM2015} $\left\vert
%TCIMACRO{\dbigcap }%
%BeginExpansion
{\displaystyle\bigcap}
%EndExpansion
\Lambda\right\vert +\left\vert
%TCIMACRO{\dbigcup }%
%BeginExpansion
{\displaystyle\bigcup}
%EndExpansion
\Lambda\right\vert =2\alpha\left(  G\right)  $ holds for every family
$\Lambda\subseteq\Omega(G)$, $\left\vert \Lambda\right\vert \geq1$;

\emph{(ii)} \cite{LevManLemma2011} $\left\vert \mathrm{corona}(G)\right\vert
+\left\vert \mathrm{core}(G)\right\vert =2\alpha\left(  G\right)  $.
\end{theorem}

\begin{proof}
\emph{(i)} By Theorems \ref{th5} and \ref{th4}\emph{(ii)}, both $%
%TCIMACRO{\dbigcup }%
%BeginExpansion
{\displaystyle\bigcup}
%EndExpansion
\Lambda$ and $%
%TCIMACRO{\dbigcap }%
%BeginExpansion
{\displaystyle\bigcap}
%EndExpansion
\Lambda$ are critical sets. According to Lemma \ref{lem2}, we have
\begin{align*}
d\left(
%TCIMACRO{\dbigcup }%
%BeginExpansion
{\displaystyle\bigcup}
%EndExpansion
\Lambda\right)   &  =\left\vert
%TCIMACRO{\dbigcup }%
%BeginExpansion
{\displaystyle\bigcup}
%EndExpansion
\Lambda\right\vert -\left\vert N\left(
%TCIMACRO{\dbigcup }%
%BeginExpansion
{\displaystyle\bigcup}
%EndExpansion
\Lambda\right)  \right\vert =\alpha\left(  G\right)  -\mu\left(  G\right)  ,\\
d\left(
%TCIMACRO{\dbigcap }%
%BeginExpansion
{\displaystyle\bigcap}
%EndExpansion
\Lambda\right)   &  =\left\vert
%TCIMACRO{\dbigcap }%
%BeginExpansion
{\displaystyle\bigcap}
%EndExpansion
\Lambda\right\vert -\left\vert N\left(
%TCIMACRO{\dbigcap }%
%BeginExpansion
{\displaystyle\bigcap}
%EndExpansion
\Lambda\right)  \right\vert =\alpha\left(  G\right)  -\mu\left(  G\right)  .
\end{align*}
Hence, $\left\vert
%TCIMACRO{\dbigcap }%
%BeginExpansion
{\displaystyle\bigcap}
%EndExpansion
\Lambda\right\vert +\left\vert
%TCIMACRO{\dbigcup }%
%BeginExpansion
{\displaystyle\bigcup}
%EndExpansion
\Lambda\right\vert =2\alpha\left(  G\right)  -2\mu\left(  G\right)
+\left\vert N\left(
%TCIMACRO{\dbigcup }%
%BeginExpansion
{\displaystyle\bigcup}
%EndExpansion
\Lambda\right)  \right\vert +\left\vert N\left(
%TCIMACRO{\dbigcap }%
%BeginExpansion
{\displaystyle\bigcap}
%EndExpansion
\Lambda\right)  \right\vert $.

By Theorem \ref{Th2}\emph{(i)}, we infer that
\begin{gather*}
\left\vert N\left(
%TCIMACRO{\dbigcup }%
%BeginExpansion
{\displaystyle\bigcup}
%EndExpansion
\Lambda\right)  \right\vert +\left\vert N\left(
%TCIMACRO{\dbigcap }%
%BeginExpansion
{\displaystyle\bigcap}
%EndExpansion
\Lambda\right)  \right\vert =\left\vert N\left(
%TCIMACRO{\dbigcup }%
%BeginExpansion
{\displaystyle\bigcup}
%EndExpansion
\Lambda\right)  \right\vert +\left\vert V\left(  G\right)  \right\vert
-\left\vert
%TCIMACRO{\dbigcup }%
%BeginExpansion
{\displaystyle\bigcup}
%EndExpansion
\Lambda\right\vert =\\
=\left\vert N\left(
%TCIMACRO{\dbigcup }%
%BeginExpansion
{\displaystyle\bigcup}
%EndExpansion
\Lambda\right)  \right\vert +\alpha\left(  G\right)  +\mu\left(  G\right)
-\left\vert
%TCIMACRO{\dbigcup }%
%BeginExpansion
{\displaystyle\bigcup}
%EndExpansion
\Lambda\right\vert =\\
=\alpha\left(  G\right)  +\mu\left(  G\right)  -d\left(
%TCIMACRO{\dbigcup }%
%BeginExpansion
{\displaystyle\bigcup}
%EndExpansion
\Lambda\right)  =2\mu\left(  G\right)  .
\end{gather*}
Consequently, we obtain $\left\vert
%TCIMACRO{\dbigcap }%
%BeginExpansion
{\displaystyle\bigcap}
%EndExpansion
\Lambda\right\vert +\left\vert
%TCIMACRO{\dbigcup }%
%BeginExpansion
{\displaystyle\bigcup}
%EndExpansion
\Lambda\right\vert =2\alpha\left(  G\right)  $, as claimed.

\emph{(ii)} It follows from Part \emph{(i)}, by taking $\Lambda\subseteq
\Omega(G)$.
\end{proof}

The graph $G_{2}$ from Figure \ref{fig101101} has $\left\vert \mathrm{corona}%
(G_{2})\right\vert +\left\vert \mathrm{core}(G_{2})\right\vert =13>12=2\alpha
\left(  G_{2}\right)  $. On the other hand, there is a
non-K\"{o}nig-Egerv\'{a}ry graph, namely $G_{1}$ in Figure \ref{fig101101},
that satisfies $\left\vert \mathrm{corona}(G_{1})\right\vert +\left\vert
\mathrm{core}(G_{1})\right\vert =10=2\alpha\left(  G_{1}\right)  $.

If $%
%TCIMACRO{\dbigcap }%
%BeginExpansion
{\displaystyle\bigcap}
%EndExpansion
\Lambda$ is a critical set, then $%
%TCIMACRO{\dbigcup }%
%BeginExpansion
{\displaystyle\bigcup}
%EndExpansion
\Lambda$ is not necessarily critical. For instance, consider the graph $G$
from Figure \ref{fig144}, and $\Lambda=\{S_{1},S_{2}\}$, where $S_{1}%
=\{x,y,u\}$ and $S_{2}=\{x,y,w\}$. Clearly, $%
%TCIMACRO{\dbigcap }%
%BeginExpansion
{\displaystyle\bigcap}
%EndExpansion
\Lambda=\{x,y\}=\mathrm{core}(G)$ is critical, while $%
%TCIMACRO{\dbigcup }%
%BeginExpansion
{\displaystyle\bigcup}
%EndExpansion
\Lambda=\{x,y,u,w\}$ is not a critical set.

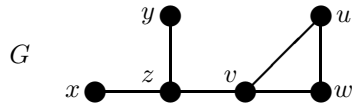
\begin{figure}[h]
\setlength{\unitlength}{1cm}\begin{picture}(5,1.2)\thicklines
\multiput(6,0)(1,0){4}{\circle*{0.29}}
\multiput(7,1)(2,0){2}{\circle*{0.29}}
\put(6,0){\line(1,0){3}}
\put(7,0){\line(0,1){1}}
\put(9,0){\line(0,1){1}}
\put(8,0){\line(1,1){1}}
\put(5.7,0){\makebox(0,0){$x$}}
\put(6.7,1){\makebox(0,0){$y$}}
\put(6.7,0.2){\makebox(0,0){$z$}}
\put(7.8,0.2){\makebox(0,0){$v$}}
\put(9.3,1){\makebox(0,0){$u$}}
\put(9.3,0){\makebox(0,0){$w$}}
\put(5,0.5){\makebox(0,0){$G$}}
\end{picture}\caption{$\mathrm{core}(G)=\ker\left(  G\right)  =\{x,y\}$.}%
\label{fig144}%
\end{figure}

\begin{theorem}
\label{th8}Let $\Lambda\subseteq\Omega(G)$, and $\left\vert \Lambda\right\vert
\geq1$. Then $G$ is a K\"{o}nig-Egerv\'{a}ry graph if and only if $%
%TCIMACRO{\dbigcup }%
%BeginExpansion
{\displaystyle\bigcup}
%EndExpansion
\Lambda$ is critical and $\left\vert
%TCIMACRO{\dbigcap }%
%BeginExpansion
{\displaystyle\bigcap}
%EndExpansion
\Lambda\right\vert +\left\vert
%TCIMACRO{\dbigcup }%
%BeginExpansion
{\displaystyle\bigcup}
%EndExpansion
\Lambda\right\vert =2\alpha\left(  G\right)  $.
\end{theorem}

\begin{proof}
Combining Theorems \ref{th5} and \ref{th4}\emph{(ii)}, we infer that $%
%TCIMACRO{\dbigcup }%
%BeginExpansion
{\displaystyle\bigcup}
%EndExpansion
\Lambda$ is critical. The equality $\left\vert
%TCIMACRO{\dbigcap }%
%BeginExpansion
{\displaystyle\bigcap}
%EndExpansion
\Lambda\right\vert +\left\vert
%TCIMACRO{\dbigcup }%
%BeginExpansion
{\displaystyle\bigcup}
%EndExpansion
\Lambda\right\vert =2\alpha\left(  G\right)  $ holds by Theorem \ref{th9}%
\emph{(i)}.

Conversely, according to Theorem \ref{Th2}\emph{(i)}, the set $%
%TCIMACRO{\dbigcap }%
%BeginExpansion
{\displaystyle\bigcap}
%EndExpansion
\Lambda$ is critical. Hence, by Theorem \ref{th2}, there exists a matching
from $N\left(
%TCIMACRO{\dbigcap }%
%BeginExpansion
{\displaystyle\bigcap}
%EndExpansion
\Lambda\right)  $ into $%
%TCIMACRO{\dbigcap }%
%BeginExpansion
{\displaystyle\bigcap}
%EndExpansion
\Lambda$. Theorem \ref{Th2}\emph{(i)} ensures that $\left\vert N\left(
%TCIMACRO{\dbigcap }%
%BeginExpansion
{\displaystyle\bigcap}
%EndExpansion
\Lambda\right)  \right\vert +\left\vert
%TCIMACRO{\dbigcup }%
%BeginExpansion
{\displaystyle\bigcup}
%EndExpansion
\Lambda\right\vert =\left\vert V\left(  G\right)  \right\vert $, which means
that $%
%TCIMACRO{\dbigcup }%
%BeginExpansion
{\displaystyle\bigcup}
%EndExpansion
\Lambda\cup$ $N\left(
%TCIMACRO{\dbigcap }%
%BeginExpansion
{\displaystyle\bigcap}
%EndExpansion
\Lambda\right)  =V(G)$. To complete the proof that $G$ is a
K\"{o}nig-Egerv\'{a}ry graph, one has to find a matching from $%
%TCIMACRO{\dbigcup }%
%BeginExpansion
{\displaystyle\bigcup}
%EndExpansion
\Lambda-S$ into $S-%
%TCIMACRO{\dbigcap }%
%BeginExpansion
{\displaystyle\bigcap}
%EndExpansion
\Lambda$ for some maximum independent set $S\in\Lambda$. Actually, in
accordance with Lemma \ref{MatchLem}, there is a matching, say $M$, from $S-%
%TCIMACRO{\dbigcap }%
%BeginExpansion
{\displaystyle\bigcap}
%EndExpansion
\Lambda$ to $%
%TCIMACRO{\dbigcup }%
%BeginExpansion
{\displaystyle\bigcup}
%EndExpansion
\Lambda-S$. Since $\left\vert
%TCIMACRO{\dbigcap }%
%BeginExpansion
{\displaystyle\bigcap}
%EndExpansion
\Lambda\right\vert +\left\vert
%TCIMACRO{\dbigcup }%
%BeginExpansion
{\displaystyle\bigcup}
%EndExpansion
\Lambda\right\vert =2\alpha\left(  G\right)  $ and $\left\vert S\right\vert
=\alpha\left(  G\right)  $, we infer that $\left\vert S-%
%TCIMACRO{\dbigcap }%
%BeginExpansion
{\displaystyle\bigcap}
%EndExpansion
\Lambda\right\vert =\left\vert
%TCIMACRO{\dbigcup }%
%BeginExpansion
{\displaystyle\bigcup}
%EndExpansion
\Lambda-S\right\vert $. Consequently, $M$ is a perfect matching, and this
shows that $M$ is also a matching from $%
%TCIMACRO{\dbigcup }%
%BeginExpansion
{\displaystyle\bigcup}
%EndExpansion
\Lambda-S$ into $S-%
%TCIMACRO{\dbigcap }%
%BeginExpansion
{\displaystyle\bigcap}
%EndExpansion
\Lambda$, as required.
\end{proof}

\begin{remark}
If $%
%TCIMACRO{\dbigcap }%
%BeginExpansion
{\displaystyle\bigcap}
%EndExpansion
\Lambda$ is critical and $\left\vert
%TCIMACRO{\dbigcap }%
%BeginExpansion
{\displaystyle\bigcap}
%EndExpansion
\Lambda\right\vert +\left\vert
%TCIMACRO{\dbigcup }%
%BeginExpansion
{\displaystyle\bigcup}
%EndExpansion
\Lambda\right\vert =2\alpha\left(  G\right)  $, then $G$ is not necessarily a
K\"{o}nig-Egerv\'{a}ry graph. For example, let $G$ be the graph from Figure
\ref{fig144}, and $\Lambda=\{S_{1},S_{2}\}$, where $S_{1}=\{x,y,u\},S_{2}%
=\{x,y,w\}$. Hence, $%
%TCIMACRO{\dbigcap }%
%BeginExpansion
{\displaystyle\bigcap}
%EndExpansion
\Lambda=\{x,y\}=\ker\left(  G\right)  $ is a critical set, $\left\vert
%TCIMACRO{\dbigcap }%
%BeginExpansion
{\displaystyle\bigcap}
%EndExpansion
\Lambda\right\vert +\left\vert
%TCIMACRO{\dbigcup }%
%BeginExpansion
{\displaystyle\bigcup}
%EndExpansion
\Lambda\right\vert =6=2\alpha\left(  G\right)  $, while $G$ is not a
K\"{o}nig-Egerv\'{a}ry graph. Clearly, the set $%
%TCIMACRO{\dbigcup }%
%BeginExpansion
{\displaystyle\bigcup}
%EndExpansion
\Lambda=\{x,y,u,w\}$ is not critical.
\end{remark}

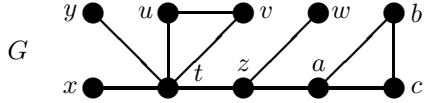
\begin{figure}[h]
\setlength{\unitlength}{1cm}\begin{picture}(5,1.2)\thicklines
\multiput(5,0)(1,0){5}{\circle*{0.29}}
\multiput(5,1)(1,0){5}{\circle*{0.29}}
\put(5,0){\line(1,0){4}}
\put(5,1){\line(1,-1){1}}
\put(6,0){\line(0,1){1}}
\put(6,0){\line(1,1){1}}
\put(6,1){\line(1,0){1}}
\put(7,0){\line(1,1){1}}
\put(8,0){\line(1,1){1}}
\put(9,0){\line(0,1){1}}
\put(4.7,0){\makebox(0,0){$x$}}
\put(4.7,1){\makebox(0,0){$y$}}
\put(6.4,0.2){\makebox(0,0){$t$}}
\put(5.7,1){\makebox(0,0){$u$}}
\put(7.3,1){\makebox(0,0){$v$}}
\put(8.3,1){\makebox(0,0){$w$}}
\put(7,0.3){\makebox(0,0){$z$}}
\put(8,0.3){\makebox(0,0){$a$}}
\put(9.3,1){\makebox(0,0){$b$}}
\put(9.3,0){\makebox(0,0){$c$}}
\put(4,0.5){\makebox(0,0){$G$}}
\end{picture}\caption{$G$ is a non-K\"{o}nig-Egerv\'{a}ry graph with
$\mathrm{core}(G)=\ker\left(  G\right)  =\{x,y\}$.}%
\label{fig41}%
\end{figure}

If $A_{1}$, $A_{2}$ are independent sets such that $A_{1}\cup A_{2}$ is
critical, then $A_{1}$ and $A_{2}$ are not necessarily critical. For instance,
consider the graph $G$ from Figure \ref{fig41}, where $A_{1}\cup
A_{2}=\{u,v,x,y\}$ is a critical set, while none of $A_{1}=\{u,x\}$ and
$A_{2}=\{v,y\}$ is critical. The case is different whenever the two
independent sets are also maximum.

\begin{corollary}
\label{cor5}The following assertions are equivalent:

\emph{(i)} $G$ is a K\"{o}nig-Egerv\'{a}ry graph;

\emph{(ii)} for every $S_{1},S_{2}\in\Omega(G)$, the set $S_{1}\cup S_{2}$ is critical;

\emph{(iii)} there exist $S_{1},S_{2}\in\Omega(G)$, such that $S_{1}\cup
S_{2}$ is critical.
\end{corollary}

\begin{proof}
\emph{(i)}\textit{ }$\Rightarrow$\textit{ }\emph{(ii) }It follows combining
Theorem \ref{th5}\emph{(iii)} and Theorem \ref{th4}\emph{(ii)}.

\emph{(ii)}\textit{ }$\Rightarrow$\textit{ }\emph{(iii) }Clear.

\emph{(iii)}\textit{ }$\Rightarrow$\textit{ }\emph{(i) }It is true according
to Theorem \ref{th8}, because
\[
\left\vert
%TCIMACRO{\dbigcap }%
%BeginExpansion
{\displaystyle\bigcap}
%EndExpansion
\Lambda\right\vert +\left\vert
%TCIMACRO{\dbigcup }%
%BeginExpansion
{\displaystyle\bigcup}
%EndExpansion
\Lambda\right\vert =\left\vert S_{1}\cup S_{2}\right\vert +\left\vert
S_{1}\cap S_{2}\right\vert =\left\vert S_{1}\right\vert +\left\vert
S_{2}\right\vert =2\alpha\left(  G\right)
\]
is automatically valid for every family $\Lambda\subseteq\Omega(G)$ with
$\left\vert \Lambda\right\vert =2$.
\end{proof}

\begin{remark}
If $G$ is a K\"{o}nig-Egerv\'{a}ry graph, then $S\cup A$ is not necessarily
critical for every $S\in\Omega(G)$ and $A\in\mathrm{Ind}(G)$. For instance,
consider the graph $G$ in Figure \ref{fig141}, and $S=\{a,b,c,d\}\in\Omega
(G)$. The sets $A_{1}=\{v\}$ and $A_{2}=\{w\}$ are independent, $S\cup A_{1}%
$\ is critical (because $N\left(  S\cup A_{1}\right)  =\{u,v,w,a\}$), while
$S\cup A_{2}$ is not critical (as $N\left(  S\cup A_{2}\right)
=\{u,v,w,c,d\}$).
\end{remark}

\begin{figure}[h]
\setlength{\unitlength}{1cm}\begin{picture}(5,1)\thicklines
\multiput(6,0)(1,0){4}{\circle*{0.29}}
\multiput(6,1)(1,0){3}{\circle*{0.29}}
\put(6,0){\line(1,0){3}}
\put(6,1){\line(1,0){1}}
\put(7,0){\line(0,1){1}}
\put(7,1){\line(1,-1){1}}
\put(8,0){\line(0,1){1}}
\put(5.7,0){\makebox(0,0){$a$}}
\put(5.7,1){\makebox(0,0){$b$}}
\put(6.75,0.2){\makebox(0,0){$v$}}
\put(7.3,){\makebox(0,0){$u$}}
\put(8.3,0.2){\makebox(0,0){$w$}}
\put(8.3,1){\makebox(0,0){$d$}}
\put(9.3,0){\makebox(0,0){$c$}}
\put(4.5,0.5){\makebox(0,0){$G$}}
\end{picture}\caption{$G$ is a K\"{o}nig-Egerv\'{a}ry graph with $d(G)=1$.}%
\label{fig141}%
\end{figure}
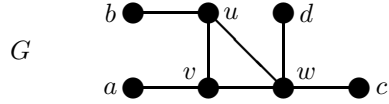

\begin{corollary}
\label{cor555}The following assertions are equivalent:

\emph{(i)} $G$ is a K\"{o}nig-Egerv\'{a}ry graph;

\emph{(ii)} for every $S\in\Omega(G)$ there exists $A\in\mathrm{Ind}(G)$, such
that the set $S\cup A$ is critical;

\emph{(iii)} there are $S\in\Omega(G)$ and $A\in\mathrm{Ind}(G)$, such that
the set $S\cup A$ is critical.
\end{corollary}

\begin{proof}
\emph{(i)}\textit{ }$\Rightarrow$\textit{ }\emph{(ii) }By Theorem
\ref{th5}\emph{(iii)}, we know that every $S\in\Omega(G)$ is critical. Hence,
$S\cup A$ is critical for any $A\subseteq S$.

\emph{(ii)}\textit{ }$\Rightarrow$\textit{ }\emph{(iii) }Clear.

\emph{(iii)}\textit{ }$\Rightarrow$\textit{ }\emph{(i) }If $A\subseteq S$, the
result follows by Theorem \ref{th5}. Otherwise, we can suppose that $S\cap
A=\emptyset$. By Theorem \ref{Berge}, we know that $\left\vert N\left(
A\right)  \cap S\right\vert \geq\left\vert A\right\vert $. Since
\[
\left\vert N\left(  S\cup A\right)  \right\vert \geq\left\vert N\left(
A\right)  \cap S\right\vert +\left\vert N\left(  S\right)  \right\vert
\geq\left\vert A\right\vert +\left\vert N\left(  S\right)  \right\vert ,
\]
we obtain
\begin{align*}
d\left(  G\right)   &  =d\left(  S\cup A\right)  =\left\vert S\cup
A\right\vert -\left\vert N\left(  S\cup A\right)  \right\vert \\
&  \leq\left(  \left\vert S\right\vert +\left\vert A\right\vert \right)
-\left(  \left\vert A\right\vert +\left\vert N\left(  S\right)  \right\vert
\right)  =d(S)\text{.}%
\end{align*}
Therefore, $d\left(  S\right)  =d(G)$, i.e., $S$ is a critical set. According
to Theorem \ref{th5}, $G$ is a K\"{o}nig-Egerv\'{a}ry graph.
\end{proof}

\section{$\mathrm{\ker}\left(  G\right)  $ and $\mathrm{diadem}(G)$ in
K\"{o}nig-Egerv\'{a}ry graphs}

\begin{theorem}
\label{th11}If $G$ is a K\"{o}nig-Egerv\'{a}ry graph, then

\emph{(i) }$\mathrm{diadem}(G)=\mathrm{corona}(G)$, while $\mathrm{diadem}%
(G)\subseteq\mathrm{corona}(G)$ is true for every graph;

\emph{(ii) }$\left\vert \ker\left(  G\right)  \right\vert +\left\vert
\mathrm{diadem}\left(  G\right)  \right\vert \leq2\alpha\left(  G\right)  $.
\end{theorem}

\begin{proof}
\emph{(i)} Every $S\in\Omega\left(  G\right)  $ is a critical set, by Theorem
\ref{th5}. Hence we deduce that $\mathrm{corona}(G)\subseteq\mathrm{diadem}%
(G)$. On the other hand, for every graph each critical independent set is
included in a maximum independent set, according to Theorem \ref{th3}. Thus,
we infer that $\mathrm{diadem}(G)\subseteq\mathrm{corona}(G)$. Consequently,
the equality $\mathrm{diadem}(G)=\mathrm{corona}(G)$ holds.

\emph{(ii)} It follows by combining Part \emph{(i)}, Theorem \ref{th9}%
\emph{(ii)} and Theorem \ref{th4}\emph{(i)}.
\end{proof}

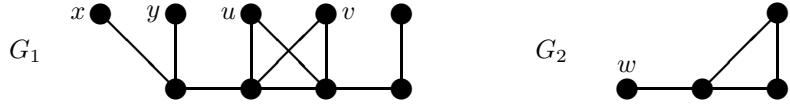
\begin{figure}[h]
\setlength{\unitlength}{1cm}\begin{picture}(5,1.3)\thicklines
\multiput(4,0)(1,0){4}{\circle*{0.29}}
\multiput(3,1)(1,0){5}{\circle*{0.29}}
\put(4,0){\line(1,0){3}}
\put(4,0){\line(0,1){1}}
\put(3,1){\line(1,-1){1}}
\put(5,0){\line(0,1){1}}
\put(5,0){\line(1,1){1}}
\put(5,1){\line(1,-1){1}}
\put(6,0){\line(0,1){1}}
\put(7,0){\line(0,1){1}}
\put(2.7,1){\makebox(0,0){$x$}}
\put(3.7,1){\makebox(0,0){$y$}}
\put(4.7,1){\makebox(0,0){$u$}}
\put(6.3,1){\makebox(0,0){$v$}}
\put(2,0.5){\makebox(0,0){$G_{1}$}}
\multiput(10,0)(1,0){2}{\circle*{0.29}}
\multiput(12,0)(0,1){2}{\circle*{0.29}}
\put(10,0){\line(1,0){2}}
\put(11,0){\line(1,1){1}}
\put(12,0){\line(0,1){1}}
\put(10,0.3){\makebox(0,0){$w$}}
\put(9,0.5){\makebox(0,0){$G_{2}$}}
\end{picture}\caption{$G_{1}$ and $G_{2}$\ are K\"{o}nig-Egerv\'{a}ry graphs.
$ker(G_{1})=\{x,y\}$ and $ker(G_{2})=\emptyset$.}%
\label{fig222}%
\end{figure}

The K\"{o}nig-Egerv\'{a}ry graphs from Figure \ref{fig222} satisfy $\left\vert
\ker\left(  G\right)  \right\vert +\left\vert \text{\textrm{diadem}}\left(
G\right)  \right\vert <2\alpha\left(  G\right)  $.

The graph $G_{1}$ from Figure \ref{fig178} is a non-bipartite
K\"{o}nig-Egerv\'{a}ry graph, such that $\mathrm{\ker}(G_{1})=\mathrm{core}%
(G_{1})$ and \textrm{diadem}$\left(  G_{1}\right)  =\mathrm{corona}(G_{1})$.
The combination of $\mathrm{diadem}(G)\varsubsetneqq\mathrm{corona}(G)$ and
$\mathrm{\ker}(G)=\mathrm{core}(G)$ is realized by the
non-K\"{o}nig-Egerv\'{a}ry graph $G_{2}$ from Figure \ref{fig178}, because
$\mathrm{\ker}(G_{2})=\mathrm{core}(G_{2})$ and \textrm{diadem}$\left(
G_{2}\right)  \cup\{z,t,v,w\}=\mathrm{corona}(G_{2})$.

\begin{figure}[h]
\setlength{\unitlength}{1cm}\begin{picture}(5,1.2)\thicklines
\multiput(2,1)(1,0){6}{\circle*{0.29}}
\multiput(4,0)(1,0){4}{\circle*{0.29}}
\put(2,0){\circle*{0.29}}
\put(2,0){\line(1,0){5}}
\put(2,1){\line(2,-1){2}}
\put(3,1){\line(1,-1){1}}
\put(3,1){\line(1,0){1}}
\put(4,0){\line(0,1){1}}
\put(4,0){\line(1,1){1}}
\put(5,1){\line(1,0){2}}
\put(6,0){\line(0,1){1}}
\put(1.7,0){\makebox(0,0){$a$}}
\put(1.7,1){\makebox(0,0){$b$}}
\put(1,0.5){\makebox(0,0){$G_{1}$}}
\multiput(9,0)(1,0){5}{\circle*{0.29}}
\multiput(10,1)(1,0){2}{\circle*{0.29}}
\put(13,1){\circle*{0.29}}
\put(9,0){\line(1,0){4}}
\put(10,0){\line(0,1){1}}
\put(11,0){\line(0,1){1}}
\put(11,1){\line(1,-1){1}}
\put(13,0){\line(0,1){1}}
\put(9,0.3){\makebox(0,0){$x$}}
\put(10.3,1){\makebox(0,0){$y$}}
\put(11.3,1){\makebox(0,0){$t$}}
\put(12.7,1){\makebox(0,0){$u$}}
\put(10.7,0.3){\makebox(0,0){$z$}}
\put(12,0.3){\makebox(0,0){$v$}}
\put(12.7,0.3){\makebox(0,0){$w$}}
\put(8.2,0.5){\makebox(0,0){$G_{2}$}}
\end{picture}\caption{$\mathrm{core}(G_{1})=\{a,b\}$ and $\mathrm{core}%
(G_{2})=\{x,y\}$.}%
\label{fig178}%
\end{figure}
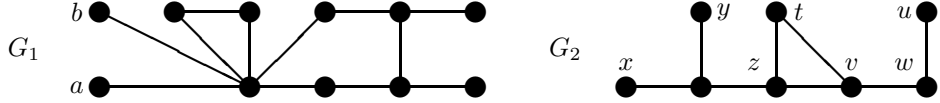

The graph $G_{2}$ from Figure \ref{fig101101} has : $\mathrm{corona}%
(G_{2})=V\left(  G_{2}\right)  -\{v_{1},v_{6}\}$ and $N(\mathrm{corona}%
(G_{2}))=V\left(  G_{2}\right)  -\{x,y,z\}$. Hence, $d(\mathrm{corona}%
(G_{2}))=1=d(G_{2})$, i.e., $\mathrm{corona}(G_{2})$ is a critical set, while
$\mathrm{diadem}(G_{2})=\{x,y,z,v_{2},v_{3},v_{5}\}\varsubsetneq
\mathrm{corona}(G_{2})$. Thus, the graph $G_{2}$ from Figure \ref{fig101101}
shows that it is possible for a graph to have $\mathrm{diadem}%
(G)\varsubsetneqq\mathrm{corona}(G)$ and $\mathrm{\ker}(G)\varsubsetneqq
\mathrm{core}(G)$. On the other hand, the graph $G_{2}$ from Figure
\ref{fig101101} gives an example where not every critical set is a subset of
$\mathrm{diadem}(G)$.

\begin{corollary}
\label{cor4}The following assertions are equivalent:

\emph{(i)} $G$ is a K\"{o}nig-Egerv\'{a}ry graph;

\emph{(ii)} $\mathrm{diadem}(G)=\mathrm{corona}(G)$ and $\left\vert
\mathrm{core}(G)\right\vert +\left\vert \mathrm{corona}(G)\right\vert
=2\alpha\left(  G\right)  $;

\emph{(iii)} $\mathrm{corona}(G)$ is a critical set and $\left\vert
\mathrm{core}(G)\right\vert +\left\vert \mathrm{corona}(G)\right\vert
=2\alpha\left(  G\right)  $.
\end{corollary}

\begin{proof}
\emph{(i)}\textit{ }$\Rightarrow$\textit{ }\emph{(ii) }It is true, by applying
Theorem \ref{th11}\emph{(i)} and Theorem \ref{th9}\emph{(ii)}.

\emph{(ii)}\textit{ }$\Rightarrow$\textit{ }\emph{(iii) }It follows from
Corollary \ref{cor3}.

\emph{(iii)}\textit{ }$\Rightarrow$\textit{ }\emph{(i) }Take $\Lambda
=\Omega(G)$ and use Theorem \ref{th8}.
\end{proof}

Notice that the graph $G_{1}$ from Figure \ref{fig101101} satisfies
$\left\vert \mathrm{core}(G_{1})\right\vert +\left\vert \mathrm{corona}%
(G_{1})\right\vert =2\alpha\left(  G_{1}\right)  $, while $d(\mathrm{corona}%
(G_{1}))=0<d(G_{1})=1$, i.e., $\mathrm{corona}(G_{1})$ is not a critical set,
because $\mathrm{corona}(G_{1})=V\left(  G_{1}\right)  -\{c,d\}$ and
$N(\mathrm{corona}(G_{1}))=V\left(  G_{1}\right)  -\{a,b\}$. On the other
hand, the graph $G_{2}$ from Figure \ref{fig101101} satisfies $\left\vert
\mathrm{core}(G_{2})\right\vert +\left\vert \mathrm{corona}(G_{2})\right\vert
=13>12=2\alpha\left(  G_{2}\right)  $, while $\mathrm{corona}(G_{2})$ is a
critical set.

\section{Conclusions}

In this paper we focus on interconnections between critical unions and
intersections of maximum independent sets, with emphasis on
K\"{o}nig-Egerv\'{a}ry graphs. In \cite{LevManLemma2011} we showed that
$2\alpha\left(  G\right)  \leq\left\vert \text{\textrm{core}}\left(  G\right)
\right\vert +\left\vert \text{\textrm{corona}}\left(  G\right)  \right\vert $
is true for every graph, while the equality \textrm{diadem}$\left(  G\right)
=\mathrm{corona}(G)$ holds for each K\"{o}nig-Egerv\'{a}ry graph $G$, by
Theorem \ref{th11}\emph{(i)}. According to Theorem \ref{th4}\emph{(i)},
$\mathrm{\ker}(G)\subseteq\mathrm{core}(G)$ for every graph. On the other
hand, Theorem \ref{th3}\emph{ }implies the inclusion \textrm{diadem}$\left(
G\right)  \subseteq\mathrm{corona}(G)$. Hence%

\[
\left\vert \ker\left(  G\right)  \right\vert +\left\vert \text{\textrm{diadem}%
}\left(  G\right)  \right\vert \leq\left\vert \text{\textrm{core}}\left(
G\right)  \right\vert +\left\vert \text{\textrm{corona}}\left(  G\right)
\right\vert
\]
for each graph $G$. These remarks together with Theorem \ref{th11}\emph{(ii)}
motivate the following.

\begin{conjecture}
$\left\vert \ker\left(  G\right)  \right\vert +\left\vert
\text{\textrm{diadem}}\left(  G\right)  \right\vert \leq2\alpha\left(
G\right)  $ is true for every graph $G$.
\end{conjecture}

When it is proved one can conclude that the following inequalities:
\[
\left\vert \ker\left(  G\right)  \right\vert +\left\vert \text{\textrm{diadem}%
}\left(  G\right)  \right\vert \leq2\alpha\left(  G\right)  \leq\left\vert
\text{\textrm{core}}\left(  G\right)  \right\vert +\left\vert
\text{\textrm{corona}}\left(  G\right)  \right\vert
\]
hold for every graph $G$.

Theorem \ref{th11} claims that $\mathrm{diadem}(G)=\mathrm{corona}(G)$ is a
necessary condition for $G$ to be a K\"{o}nig-Egerv\'{a}ry graph, while
Corollary \ref{cor4} shows that, apparently, this equality is not enough.
These facts motivate the following.

\begin{conjecture}
If\emph{ }$\mathrm{diadem}(G)=\mathrm{corona}(G)$, then $G$ is a
K\"{o}nig-Egerv\'{a}ry graph.
\end{conjecture}

The graphs in Figure \ref{fig1} are non-K\"{o}nig-Egerv\'{a}ry graphs;
$\mathrm{core}(G_{1})=\{a,b,c,d\}$ and it is a critical set, while
$\mathrm{core}(G_{2})=\{x,y,z,w\}$ and it is not critical.

\begin{figure}[h]
\setlength{\unitlength}{1cm}\begin{picture}(5,1.8)\thicklines
\multiput(1,0.5)(1,0){6}{\circle*{0.29}}
\multiput(1,1.5)(1,0){5}{\circle*{0.29}}
\put(1,0.5){\line(1,0){5}}
\put(1,1.5){\line(1,-1){1}}
\put(2,0.5){\line(0,1){1}}
\put(3,1.5){\line(1,-1){1}}
\put(3,1.5){\line(1,0){1}}
\put(4,0.5){\line(0,1){1}}
\put(5,1.5){\line(1,-1){1}}
\put(5,0.5){\line(0,1){1}}
\put(0.7,1.5){\makebox(0,0){$a$}}
\put(0.7,0.5){\makebox(0,0){$b$}}
\put(2.3,1.5){\makebox(0,0){$c$}}
\put(3,0.85){\makebox(0,0){$d$}}
\put(3,0){\makebox(0,0){$G_{1}$}}
\multiput(7,0.5)(1,0){7}{\circle*{0.29}}
\multiput(7,1.5)(1,0){6}{\circle*{0.29}}
\put(7,0.5){\line(1,0){6}}
\put(7,1.5){\line(1,-1){1}}
\put(8,0.5){\line(0,1){1}}
\put(9,0.5){\line(0,1){1}}
\put(9,0.5){\line(1,1){1}}
\put(9,1.5){\line(1,0){1}}
\put(11,0.5){\line(0,1){1}}
\put(12,1.5){\line(1,-1){1}}
\put(11,1.5){\line(1,0){1}}
\put(6.7,1.5){\makebox(0,0){$x$}}
\put(6.7,0.5){\makebox(0,0){$y$}}
\put(8.3,1.5){\makebox(0,0){$z$}}
\put(10,0.8){\makebox(0,0){$w$}}
\put(10,0){\makebox(0,0){$G_{2}$}}
\end{picture}\caption{Both $G_{1}$ and $G_{2}$\ are non-K\"{o}nig-Egerv\'{a}ry
graphs.}%
\label{fig1}%
\end{figure}
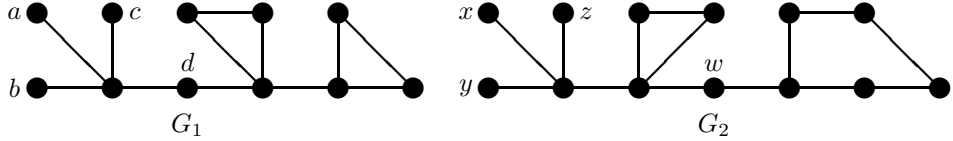By Corollary \ref{cor2}, $\mathrm{core}(G)$ is a critical set for
every K\"{o}nig-Egerv\'{a}ry graph. It justifies the following.

\begin{problem}
Characterize graphs such that $\mathrm{core}(G)$ is a critical set.
\end{problem}

It is known that the sets $\mathrm{\ker}(G)$ and $\mathrm{core}(G)$ coincide
for bipartite graphs \cite{LevMan2011b}. Notice that there are non-bipartite
graphs enjoying the equality $\mathrm{\ker}(G)=\mathrm{core}(G)$; e.g., the
graphs from Figure \ref{fig14}, where only $G_{1}$ is a K\"{o}nig-Egerv\'{a}ry
graph. \begin{figure}[h]
\setlength{\unitlength}{1cm}\begin{picture}(5,1.2)\thicklines
\multiput(2,0)(1,0){4}{\circle*{0.29}}
\multiput(3,1)(1,0){3}{\circle*{0.29}}
\put(2,0){\line(1,0){3}}
\put(3,0){\line(0,1){1}}
\put(5,0){\line(0,1){1}}
\put(4,1){\line(1,0){1}}
\put(3,0){\line(1,1){1}}
\put(1.7,0){\makebox(0,0){$x$}}
\put(2.7,1){\makebox(0,0){$y$}}
\put(1,0.5){\makebox(0,0){$G_{1}$}}
\multiput(8,0)(1,0){5}{\circle*{0.29}}
\multiput(9,1)(1,0){3}{\circle*{0.29}}
\put(8,0){\line(1,0){4}}
\put(9,0){\line(0,1){1}}
\put(10,0){\line(0,1){1}}
\put(10,1){\line(1,0){1}}
\put(11,1){\line(1,-1){1}}
\put(7.7,0){\makebox(0,0){$a$}}
\put(8.7,1){\makebox(0,0){$b$}}
\put(7,0.5){\makebox(0,0){$G_{2}$}}
\end{picture}\caption{$\mathrm{core}(G_{1})=\ker\left(  G_{1}\right)
=\{x,y\}$ and $\mathrm{core}(G_{2})=\ker\left(  G_{2}\right)  =\{a,b\}$.}%
\label{fig14}%
\end{figure}
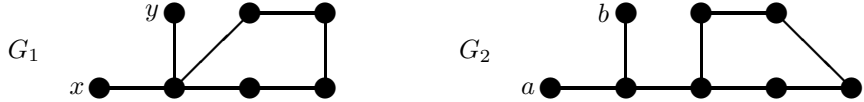

There is a non-bipartite K\"{o}nig-Egerv\'{a}ry graph $G$, such that
$\mathrm{\ker}(G)\neq\mathrm{core}(G)$. For instance, the graph $G_{1}$ from
Figure \ref{fig222} has $\mathrm{\ker}(G_{1})=\left\{  x,y\right\}  $, while
$\mathrm{core}(G_{1})=\left\{  x,y,u,v\right\}  $. The graph $G_{2}$ from
Figure \ref{fig222} has $\mathrm{\ker}(G_{2})=\emptyset$, while $\mathrm{core}%
(G_{2})=\left\{  w\right\}  $. We propose the following.

\begin{problem}
Characterize (K\"{o}nig-Egerv\'{a}ry) graphs satisfying $\ker\left(  G\right)
=\mathrm{core}(G)$.
\end{problem}

\end{document}